\documentclass[11pt, reqno]{amsart}
\usepackage[english]{babel}
\usepackage{physics}
\usepackage{braket}
\usepackage{dsfont}
\usepackage[table,xcdraw]{xcolor}
\usepackage{multirow}
\usepackage{graphicx}
\usepackage{tikz}
\usepackage{amssymb}
\usepackage{fontawesome}
\usetikzlibrary{matrix}
\usepackage{stmaryrd}
\usepackage{makecell}

\usepackage[margin=1in]{geometry}
\usepackage[pagebackref, colorlinks = true, linkcolor = blue, urlcolor  = blue, citecolor = red]{hyperref}

\renewcommand{\epsilon}{\varepsilon}
\renewcommand{\phi}{\varphi}

\makeatletter
\newtheorem*{rep@theorem}{\rep@title}
\newcommand{\newreptheorem}[2]{
\newenvironment{rep#1}[1]{
 \def\rep@title{#2 \ref{##1}}
 \begin{rep@theorem}}
 {\end{rep@theorem}}}
\makeatother

\newmuskip\pFqmuskip

\newcommand*\pFq[6][8]{
  \begingroup
  \pFqmuskip=#1mu\relax
  \mathchardef\normalcomma=\mathcode`,
  \mathcode`\,=\string"8000
  \begingroup\lccode`\~=`\,
  \lowercase{\endgroup\let~}\pFqcomma
  {}_{#2}F_{#3}{\left[\genfrac..{0pt}{}{#4}{#5};#6\right]}
  \endgroup
}
\newcommand{\pFqcomma}{{\normalcomma}\mskip\pFqmuskip}

\newtheorem{thm}{Theorem}[section]
\newtheorem*{thm*}{Theorem}
\newreptheorem{thm}{Theorem}
\newtheorem{lem}[thm]{Lemma}
\newtheorem{cor}[thm]{Corollary}
\newtheorem{ex}[thm]{Example}
\newtheorem{defi}[thm]{Definition}
\newtheorem{prop}[thm]{Proposition}
\newtheorem{remark}[thm]{Remark}

\begin{document}

\author{Andreas Bluhm}
\email{bluhm@math.ku.dk}
\address{QMATH, Department of Mathematical Sciences, University of Copenhagen, Universitetsparken 5, 2100 Copenhagen, Denmark}

\author{Ion Nechita}
\email{ion.nechita@univ-tlse3.fr}
\address{Laboratoire de Physique Th\'eorique, Universit\'e de Toulouse, CNRS, UPS, France}

\title{A tensor norm approach to quantum compatibility}

\begin{abstract}
Measurement incompatibility is one of the most striking examples of how quantum physics is different from classical physics. Two measurements are incompatible if they cannot arise via classical post-processing from a third one. A natural way to quantify incompatibility is in terms of noise robustness. In the present article, we review recent results on the maximal noise robustness of incompatible measurements, which have been obtained by the present authors using free spectrahedra, and rederive them using tensor norms. In this way, we make them accessible to a broader audience from quantum information theory and mathematical physics and contribute to the fruitful interactions between Banach space theory and quantum information theory. We also describe incompatibility witnesses using tensor norm and matrix convex set duality, emphasizing the relation between the different notions of witnesses. 
\end{abstract}

\date{\today}

\maketitle

\setcounter{tocdepth}{1}
\tableofcontents

\section{Introduction}
The existence of incompatible observables is one of the most striking ways in which quantum mechanics differs from classical physics. Observables are incompatible if they cannot be measured at the same time \cite{Heisenberg1927, Bohr1928}. Arguably the best example of such observables are position and momentum, but interesting situations also occur in finite dimensions \cite{Heinosaari2016,guhne2021incompatible}. More precisely, quantum mechanics allows the existence of incompatible measurements, i.e., measurements which do not arise as marginals from a joint measurement. Equivalently, measurements are incompatible if there does not exist another measurement from which all of the outcomes can be obtained via classical post-processing \cite{Heinosaari2016}. The existence of incompatible measurements is in fact indispensable for the violation of Bell inequalities \cite{Fine1982} and has therefore practical implications for information processing \cite{Brunner2014}. Incompatibility thus plays a similar role to entanglement inasmuch as it can be seen as a resource for quantum information tasks  \cite{Heinosaari2015}. 

In the same way as for entanglement, measurement incompatibility will vanish given enough noise  \cite{busch2013comparing}. Hence, it is desirable to quantify the maximal noise robustness of incompatibility for a given set of physical parameters such as the dimension of the quantum system, as it will tell us how much noise we can tolerate before any possible advantage a quantum device might have due to the use of incompatible measurements will disappear. This question has been investigated in \cite{busch2013comparing, Gudder2013} and recently also in  \cite{bluhm2018joint, bluhm2020compatibility, Bluhm2020GPT}, where a connection to non-commutative convex geometry has been established. 

This work is in part a review of the latter line of work. However, in the spirit of \cite{Bluhm2020GPT}, we will derive our results using tensor norms of certain Banach spaces instead of relying on free spectrahedra as \cite{bluhm2018joint, bluhm2020compatibility}. Tensor norms have recently also been successfully used to characterize quantum phenomena such as entanglement \cite{jivulescu2020multipartite} and quantum steering \cite{jencova2022assemblages}. Unlike \cite{Bluhm2020GPT}, we will focus on quantum mechanics instead of working in the broader framework of general probabilistic theories. While most of the results of this review are already present in \cite{bluhm2018joint, bluhm2020compatibility, Bluhm2020GPT}, their explicit derivation in terms of tensor norms and convex optimization problems is novel. In this way, the present work strengthens the fruitful interaction between Banach space theory and quantum information theory. Moreover, it will make the results of \cite{bluhm2018joint, bluhm2020compatibility} accessible to a broader audience from quantum information theory and mathematical physics since it eliminates the need for background knowledge regarding free spectrahedra, making use of more familiar objects such as norms.

After an exposition of the necessary preliminaries concerning measurement incompatibility, tensor norms, and matrix convex sets in Section \ref{sec:prelminaries}, we make the first connection between measurement compatibility and tensor norms in Section \ref{sec:compatibility-norms}. Then, we go on to study incompatibility witnesses and make a second connection to tensor norms in Section \ref{sec:witness-norms}. We close with a discussion and an outlook in Section \ref{sec:final-section}.

\section{Preliminaries} \label{sec:prelminaries}

We gather in this section the basic notation used in the paper, as well as the main definitions and results regarding incompatibility in quantum mechanics, (tensor) norms, and matrix convex sets.

\subsection{Notation}
Let $n \in \mathbb N$. For simplicity, we write $[n]:=\{1, \ldots, n\}$. By $\{e_j\}_{j \in [g]}$, we denote the standard basis of $\mathbb R^g$ for some $g\in \mathbb N$. We write $\mathcal M(\mathbb C)_n$ for the $n \times n$ matrices with complex entries and $\mathcal M(\mathbb C)_n^{\mathrm{sa}}$ when restricting to Hermitian elements. $I_d$ will be the identity matrix in $d$-dimensions, where we will drop the $d$ when no confusion can arise. We will write
\begin{equation*}
    \mathcal S(\mathbb C^d) := \{\rho \in \mathcal M(\mathbb C)_d^{\mathrm{sa}}:~\rho\geq 0,~ \mathrm{Tr}[\rho] = 1\}
\end{equation*}for the set of density matrices. Let $v \in \mathbb R^g$. Then, we write for $p \in \mathbb N$
\begin{equation*}
    \|v\|_p = \left(\sum_{i = 1}^g |v_i|^p\right)^{\frac{1}{p}}, \qquad \|v\|_\infty = \max_{i \in [g]} |v_i|,
\end{equation*}
and $\ell_p^g$ for the corresponding Banach space $(\mathbb R^g, \|\cdot\|_p)$. We will write $B_{\ell^g_p}$ for the respective unit balls of this norm. For $M \in \mathcal M(\mathbb C)_n$, we write for the Schatten p-norms
\begin{equation*}
    \|M\|_p = \mathrm{Tr}[|M|^p]^{\frac{1}{p}}, \qquad \|M\|_\infty = \max_{\|v\|_2 = 1} \|Mv\|_2,
\end{equation*}
and $S_p^d$ for the (real) Banach space $(\mathcal M(\mathbb C)^{\mathrm{sa}}_n, \|\cdot\|_p)$. Finally, we will write $\lambda_{\max}(M)$ for the largest eigenvalue of $M$.

\subsection{Measurement incompatibility}
In this section, we will give a short introduction to measurement incompatibility. For an introduction to the mathematics of quantum mechanics, see e.g.\ \cite{Heinosaari2011} or \cite{watrous2018theory}. Quantum mechanical measurements are described using effect operators, i.e.\
\begin{equation*}
\mathrm{Eff}_d := \Set{E \in \mathcal M(\mathbb C)_d^{\mathrm{sa}}: 0 \leq E \leq I}.
\end{equation*} 
A measurement then corresponds to a \emph{positive operator valued measure} (POVM). Let $\Sigma$ be the set of measurement outcomes, which we assume to be finite for simplicity. The corresponding POVM is then a set of effects $\Set{E_j}_{j \in \Sigma}$, $E_j \in \mathrm{Eff}_d$ for all $j \in \Sigma$, such that
\begin{equation*}
\sum_{j \in \Sigma} E_j = I_d.
\end{equation*}
Since the actual measurement outcomes are not important for us, we will write $\Sigma = [k]$ for some $k \in \mathbb N$.

Now we can define the notion of \emph{joint measurability} or \emph{compatibility} of measurements. A collection of POVMs is compatible if they arise as marginals from a joint POVM (see \cite{Heinosaari2016} for an introduction). 
\begin{defi}[Jointly measurable POVMs] \label{def:jointPOVM}
Let $\Set{E_j^{(i)}}_{j \in [k_i]}$ be a collection of $d$-dimensional POVMs, where $k_i \in \mathbb N$ for all $i \in [g]$, $d$, $g \in \mathbb N$. The POVMs are \emph{compatible} if there is a $d$-dimensional joint POVM $\Set{R_{j_1, \ldots, j_g}}$ with $j_i \in [k_i]$ such that for all $u \in [g]$ and $v \in [k_u]$,
\begin{equation*}
E_v^{(u)} = \sum_{\substack{j_i \in [k_i] \\ i \in [g] \setminus \Set{u}}} R_{j_1, \ldots, j_{u-1},v,j_{u+1}, \ldots j_g}.
\end{equation*}
\end{defi}

There is an equivalent definition of joint measurability \cite[Equation 16]{Heinosaari2016}, formulated in terms of classical post-processing, which admits a more straightforward interpretation. Measurements are compatible if and only if they arise through classical post-processing of the outcomes of a single common measurement. 
\begin{lem} \label{lem:post-processing}
Let $E^{(i)} \in (\mathcal M(\mathbb C)_d^{\mathrm{sa}})^{k_i}$, $i \in [g]$, be a collection of POVMs. These POVMs are jointly measurable if and only if there is some $m \in \mathbb N$ and a POVM $M \in (\mathcal M(\mathbb C)_d^{\mathrm{sa}})^{m}$ such that
\begin{equation*}
E^{(i)}_{j} = \sum_{x=1}^m p_i(j|x)M_x
\end{equation*}
for all $j \in [k_i]$, $i \in [g]$ and some conditional probabilities $p_i(j|x)$.
\end{lem}
For a collection of dichotomic measurements $\{E_i, I - E_i\}$, $i \in [g]$, we will often say that the effects $\{E_i\}_{i \in [g]}$ are compatible, since they completely determine the corresponding measurements. We will focus on dichotomic measurements in this work.

Not all measurements in quantum mechanics are compatible. An example of incompatible measurements are $\{P, I-P\}$ and $\{Q, I-Q\}$, where $P$, $Q$ are non-commuting orthogonal projections. However, any collection of measurements can be made compatible if we add enough noise. By adding noise we mean taking the convex combination of a POVM and a trivial measurement, i.e.\ a POVM in which all effects are proportional to the identity. These measurements are called trivial, because they do not depend on the quantum state of the system that we measure. Instead, one just generates a random outcome with probabilities given by the weights appearing in the effects of the measurement. We will focus on a specific type of noise, namely \emph{white noise}, which is a balanced form of noise:
\begin{defi}\label{def:Gamma}
Let $\mathbf k \in \mathbb N^g$, $d$, $g \in \mathbb N$. Then, we call
\begin{equation*}
\Gamma(g,d, \mathbf k) := \Set{s \in [0,1]^g: \{s_{i} E^{({i})}_j + (1-s_{i})I/k_{i}\}_{j \in [k_i]} \mathrm{~compatible~}\forall \mathrm{~POVMs~}E^{({i})} \in (\mathcal M(\mathbb C)_d^{\mathrm{sa}})^{k_{i}}}
\end{equation*}
the \emph{compatibility region} for $g$ POVMs in $d$ dimensions with $k_i$ outcomes, $i \in [g]$. If $\mathbf k = 2^{\times g}$, then we just write $\Gamma(g,d)$.
\end{defi}
Thus, the compatibility region describes the amount of noise that makes any collection of $g$ measurements with $\mathbf k$ outcomes compatible in fixed dimension $d$. Hence, it tells quantifies the maximal amount of incompatibility available for these fixed parameters. Since the convex combination of two POVMs is again a POVM, the set $\Gamma(g,d, \mathbf k)$ is convex. We list next some straightforward containments of compatibility regions with different parameters. 
\begin{prop}
Let $g$, $d$, $d^\prime \in \mathbb N$. Moreover, let $\mathbf k$, $\mathbf k^\prime \in \mathbb N^g$. Then,
\begin{enumerate}
    \item For $d^\prime \geq d$, it holds that $\Gamma(g,d^\prime, \mathbf k) \subseteq \Gamma(g,d, \mathbf k)$.
    \item For $\mathbf k^\prime \geq \mathbf k$ (i.e.\ $k_i^\prime \geq k_i$ for all $i \in [g]$), it holds that $\Gamma(g,d, \mathbf k^\prime) \subseteq \Gamma(g,d, \mathbf k)$.
    \item $\{s \in [0,1]^g: s_1 + \ldots + s_g \leq 1\}\subseteq \Gamma(g,d, \mathbf k)$.
\end{enumerate}
\end{prop}
The first assertion follows, because we can embed POVMs of a certain dimension into a higher dimension. For details see \cite[Proposition 3.6]{bluhm2018joint}. The second assertion follows because we interpret POVMs with a certain number of outcomes as POVMs with more outcomes, but for which some outcomes occur with probability $0$ for all states. See \cite[Proposition 3.35]{bluhm2020compatibility} for a formal proof. The last point of the assertion follows by convexity, since $e_i \in \Gamma(g,d, \mathbf k)$ for all $i \in [g]$, because any single measurement is compatible with any number of trivial measurements. A physical way to think about this is the following \cite{Heinosaari2016}: Given a quantum state, $(s_1, \ldots, s_g)$ defines a random variable which selects the measurement which is implemented. For all other measurements, an outcome is generated by picking an outcome uniformly at random. This procedure then defines the joint measurement, rendering the noisy measurements defined by $(s_1, \ldots, s_g)$  compatible.

Less straightforward are the following results, which have appeared in the literature:
\begin{prop} \label{prop:gamma-known-before}
Let $g$, $d \in \mathbb N$. Then,
\begin{enumerate}
    \item It holds that $\Gamma(2,d)=\{s \in [0,1]^2: \|s\|_2 \leq 1\}$.
    \item It holds that $\Gamma(3,2)=\{s \in [0,1]^3: \|s\|_2 \leq 1\}$.
\end{enumerate}
\end{prop}
The first result follows from the results of \cite{busch2013comparing} together with \cite[Proposition 3]{busch2008approximate}, while the second result can be inferred from \cite{busch1986unsharp, Brougham2007, Pal2011}.

\subsection{Tensor norms}
In this section, we will recap the basics of tensor products of Banach spaces. We refer the reader to \cite{Ryan2002} for a good introduction. Let $X$, $Y$ be two Banach spaces with norms $\norm{\cdot}_X$ and $\norm{\cdot}_Y$, respectively. Let $X^\ast$ and $Y^\ast$ be the dual spaces of $X$ and $Y$, respectively, and let their norms be $\norm{\cdot}_{X^\ast}$ and $\norm{\cdot}_{Y^\ast}$. The question is now which norm to put on the tensor product of these two vector spaces, $X \otimes Y$, to make it into a Banach space. Usually, there are infinitely many norms that one could use. To restrict this multitude of possibilities, it is natural to require that the norm on the tensor product $X \otimes Y$ behaves nicely with respect to pure tensors. This motivates the notion of \emph{reasonable crossnorms}:

\begin{defi}[\cite{Ryan2002}]
Let $X$ and $Y$ be two Banach spaces. We say that a norm $\norm{\cdot}_\alpha$ on $X \otimes Y$ is a \emph{reasonable crossnorm} if it has the following properties:
\begin{enumerate}
\item $\norm{x \otimes y}_\alpha \leq \norm{x}_X \norm{y}_Y$ for all $x \in X$, $y \in Y$,
\item For all $\phi \in X^\ast$, for all $\psi \in Y^\ast$, $\phi \otimes \psi$ is bounded on $X \otimes Y$ and $\norm{\phi \otimes \psi}_{\alpha^\ast} \leq \norm{\phi}_{X^\ast} \norm{\psi}_{Y^\ast}$,
\end{enumerate} 
where $\norm{\cdot}_{\alpha^\ast}$ is the dual norm to $\norm{\cdot}_\alpha$.
\end{defi}
There are two examples of reasonable crossnorms which can be defined on any pair of Banach spaces, the \emph{injective} and the \emph{projective} tensor norm.

\begin{defi}[Projective tensor norm]
The \emph{projective norm} of an element $z \in X \otimes Y$ is defined as
\begin{equation*}
    \norm{z}_{X \otimes_\pi Y}:= \inf \left\{\sum_i \norm{x_i}_X \norm{y_i}_Y: z = \sum_{i} x_i \otimes y_i \right\}.
\end{equation*}
\end{defi}

\begin{defi}[Injective tensor norm]
Let $z = \sum_i x_i \otimes y_i \in X \otimes Y$. Then, its \emph{injective norm} is 
\begin{equation*}
      \norm{z}_{X \otimes_\epsilon Y}:= \sup \left\{\left| \sum_i \phi(x_i) \psi(y_i)\right|: \norm{\phi}_{X^\ast} \leq 1, \norm{\psi}_{Y^\ast} \leq 1   \right\}.
\end{equation*}
\end{defi}
Importantly, $\norm{\cdot}_\epsilon$ and $\norm{\cdot}_\pi$ are dual norms, i.e.
\begin{equation*}
    \norm{z}_{X \otimes_\epsilon Y} = \sup_{\norm{\phi}_{X^\ast \otimes_\pi Y^\ast} \leq 1} |\phi(z)|
\end{equation*}
and vice versa.

The injective and projective norms are of special interest because they constitute the smallest and largest reasonable crossnorms we can put on $X \otimes Y$, respectively:
\begin{prop}[{\cite[Proposition 6.1]{Ryan2002}}] \label{prop:reasonable-cross-norms}
Let $X$ and $Y$ be Banach spaces. 
\begin{enumerate}
\item[(a)] A norm $\norm{\cdot}_\alpha$ on $X \otimes Y$ is a reasonable crossnorm if and only if
\begin{equation*}
\norm{z}_{X \otimes_\epsilon Y} \leq \norm{z}_\alpha \leq \norm{z}_{X \otimes_\pi Y}
\end{equation*}
for all $z \in X \otimes Y$
\item[(b)] If $\norm{\cdot}_\alpha$ is a reasonable crossnorm on $X \otimes Y$, then $\norm{x \otimes y}_\alpha = \norm{x}_X \norm{y}_Y$ for every $x \in X$ and every $y \in Y$. Furthermore, for all $\phi \in X^\ast$ and all $\psi \in Y^\ast$, the norm $\norm{\cdot}_{\alpha^\ast}$ satisfies $\norm{\phi \otimes \psi}_{\alpha^\ast} = \norm{\phi}_{X^\ast} \norm{\psi}_{Y^\ast}$.
\end{enumerate}
\end{prop}

We conclude with some examples:
\begin{ex}
The projective norm $\ell^{g}_1 \otimes_\pi X$ of a vector $z = \sum_{i=1}^g e_i \otimes z_i \in \mathbb R^g \otimes X$ is given by (see e.g.\ \cite[Example 2.6]{Ryan2002})
\begin{equation} \label{eq:projective-norm}
    \|z\|_{\ell_1^g \otimes_\pi X} =  \sum_{i=1}^g \|z_i\|_{X}.
\end{equation}
The injective norm $\ell^g_1 \otimes_\epsilon X$ of the same $z$ is (see e.g.\ \cite[Example 3.4]{Ryan2002})
\begin{equation} \label{eq:injective-norm}
\|z\|_{\ell_1^g \otimes_\epsilon X} = \sup_{\norm{y}_{X^\ast} \leq 1} \sum_{i = 1}^g |\langle y, z_i \rangle| = \sup_{\epsilon \in \{\pm 1\}^g} \norm{\sum_{i = 1}^g \epsilon_i z_i}_X.
\end{equation}
\end{ex}

\subsection{Matrix convex sets}

In this section, we will review some basic results from the theory of matrix convex sets. More background can be found in \cite{davidson2016dilations}, for example. We will write $\mathrm{UCP}(\mathcal B(\mathcal H), \mathcal B(\mathcal K))$ for the set of unital completely positive maps from the bounded operators on a Hilbert space $\mathcal H$ to bounded operators on a Hilbert space $\mathcal K$.

\begin{defi}
Let $g \in \mathbb N$. Moreover, let $\mathcal F_n \subseteq (\mathcal M(\mathbb C)_n^{\mathrm{sa}})^g$ for all $n \in \mathbb N$. Then, we call $\mathcal F = \bigsqcup_{n \in \mathbb N} \mathcal F_n$ a \emph{free set}. Moreover, $\mathcal F$ is a \emph{matrix convex set} if it satisfies the following two properties for any $m$, $n \in \mathbb N$:
\begin{enumerate}
\item If $X = (X_1, \ldots, X_g) \in \mathcal F_m$, $Y = (Y_1, \ldots, Y_g)\in \mathcal F_n$, then $X \oplus Y := (X_1 \oplus Y_1, \ldots, X_g \oplus Y_g) \in \mathcal F_{m+n}$
\item If $X = (X_1, \ldots, X_g) \in \mathcal F_m$ and $\Psi: \mathcal M(\mathbb C)_m \to \mathcal M(\mathbb C)_n$ is a unital completely positive (UCP) map, then $(\Psi(X_1), \ldots, \Psi(X_g)) \in \mathcal F_n$.
\end{enumerate}
That is, a matrix convex set is a free set closed under direct sums and UCP maps.
\end{defi}
In particular, it follows from the definition that all sets $\mathcal F_n$ are convex. A matrix convex set $\mathcal F$ is open/closed/bounded if all $\mathcal F_n$ defining it have this property. 

Let $\mathcal C \subseteq \mathbb R^g$ be a convex set. In general, there are infinitely many matrix convex sets $\mathcal F$ with $\mathcal F_1 = \mathcal C$. However, we can find a maximal and a minimal matrix convex set equal to $\mathcal C$ at the first level. We start with the definition of the maximal matrix convex set \cite[Definition 4.1]{davidson2016dilations}:
\begin{align*}
\label{eq:Wmax}&\mathcal W^{\max}_n(\mathcal C) :=\\
& \Set{X \in (\mathcal M(\mathbb C)_n^{\mathrm{sa}})^g: \sum_{i = 1}^g c_i X_i \leq \alpha I, \quad \forall \, c \in \mathbb R^g, \forall \alpha \in \mathbb R~\mathrm{\ s.t.\ }~\mathcal C \subseteq \Set{x \in \mathbb R^g: \langle c, x\rangle \leq \alpha}}.
\end{align*}
Note that $\mathcal W^{\max}_1(\mathcal C) = \mathcal C$, as claimed above.

We can now go on to define the minimal matrix convex set associated with $\mathcal C$. We use the definition given in \cite[Eq.~(1.4)]{passer2018minimal}:
\begin{equation*}
    \mathcal W^{\min}_n(\mathcal C) := \Set{\sum_{j} z_j \otimes X_j \in (\mathcal M(\mathbb C)_n^{\mathrm{sa}})^g: z_j \in \mathcal C~ \forall j, X_j \geq 0~\forall j, \sum_{j} X_j =I}.
\end{equation*}
Note that if $\mathcal C$ is a polytope, i.e.~it has finitely many extreme points, the number of terms in the decomposition above can be taken to be the number of
extreme points of $\mathcal C$; see the discussion after Eq.~\eqref{eq:c-norm-SDP}.
 
In this work, we will be concerned with inclusion constants, i.e.\ constants for which the inclusion
\begin{equation*}
 s\cdot \mathcal W^{\max}(\mathcal C) \subseteq \mathcal W^{\min}(\mathcal C)
\end{equation*}
holds. Here, the \emph{(asymmetrically) scaled} matrix convex set is 
\begin{equation*}
s \cdot \mathcal W^{\max}(\mathcal C):= \Set{(s_1 X_1, \ldots, s_g X_g): X \in \mathcal W^{\max}(\mathcal C)}.
\end{equation*}
\begin{defi}\label{def:Delta}
Let $d$, $g \in \mathbb N$ and $\mathcal C \subset \mathbb R^g$. The \emph{inclusion set} is defined as 
\begin{equation*}
\Delta_{\mathcal C}(d):=\Set{s \in [0,1]^g:  s \cdot \mathcal W^{\max}_d(\mathcal C) \subseteq  \mathcal W^{\min}_d(\mathcal C)}.
\end{equation*}
If $\mathcal C$ is the $\ell_\infty^g$ (resp.~the $\ell_1^g$) 
unit ball, we will write $\Delta_{\square}(g,d)$ (resp.~$\Delta_{\diamond}(g,d)$) instead of $\Delta_{\mathcal C}(d)$.
\end{defi}
Note that $\Delta_{\mathcal C}(d)$ is a convex set, because both $W_d^{\min}(\mathcal C)$ and $W_d^{\max}(\mathcal C)$ are.

We will also use the notion of dual matrix convex sets in this article, see \cite[Section 3]{davidson2016dilations}. 
\begin{defi}
Let $g\in \mathbb N$ and let $\mathcal F \subseteq \bigsqcup_{n \in \mathbb N}(\mathcal M(\mathbb C)_n^{\mathrm{sa}})^g$ be a matrix convex set. Then, its \emph{polar dual} $\mathcal F^\bullet$ is defined as
\begin{equation*}
    \mathcal F^\bullet_d := \left\{X \in (\mathcal M(\mathbb C)_d^{\mathrm{sa}})^g:  \sum_{i = 1}^g X_i \otimes F_i \leq I~\forall F \in \mathcal F\right\}
\end{equation*}
for all $d \in \mathbb N$.
\end{defi}
Note that if $0 \in \mathcal F$, then $\mathcal F^{\bullet \bullet} = \mathcal F$ \cite[Lemma 3.2]{davidson2016dilations}. It turns out that the minimal and maximal matrix convex sets are polar duals. We collect this fact in the following Lemma, which appears as Theorem 4.7 in \cite{davidson2016dilations}.
\begin{lem}\label{lem:mcs-dual}
Let $\mathcal C \subset \mathbb R^g$ be a closed convex set. Then $\mathcal W^{\min}(\mathcal C)^\bullet = \mathcal W^{\max}(\mathcal C^\circ)$, where $\mathcal C^\circ$ is the polar convex set of $\mathcal C$. If $0 \in \mathcal C$, then $\mathcal W^{\max}(\mathcal C)^\bullet = \mathcal W^{\min}(\mathcal C^\circ)$.
\end{lem}

In the remainder of this section, we will demonstrate that Definition \ref{def:Delta} for the inclusion set agrees with the definition used in \cite{bluhm2018joint, bluhm2020compatibility}. If the reader is not interested in this fact, the rest of the section may be skipped. The desired fact will follow from the next proposition.

\begin{prop}
Let $g$, $d \in \mathbb N$ and let $\mathcal F\subseteq \bigsqcup_{n \in \mathbb N}(\mathcal M(\mathbb C)_n^{\mathrm{sa}})^g$ be a matrix convex set with $0 \in \mathcal F_1$. Then, $X \in \mathcal F^{\bullet}_d$ if and only if
\begin{equation*}
    \sum_{i = 1}^g X_i \otimes F_i \leq I_{d^2} \qquad \forall (F_1, \ldots, F_g) \in \mathcal F_d.
\end{equation*}
\end{prop}
\begin{proof}
Since  $X \in \mathcal F^{\bullet}_d$ if and only if
\begin{equation*}
    \sum_{i = 1}^g X_i \otimes F_i \leq I \qquad \forall F \in \mathcal F
\end{equation*}
by definition, the first direction follows. For the converse, let us consider $(F_1, \ldots F_g) \in \mathcal F_n$. We can assume that $n > d$, because we can embed smaller matrices in dimension $d$ by adding zeroes. Let $Z = I - \sum_{i = 1}^g X_i \otimes F_i$. Then, 
\begin{equation*}
    \sum_{i = 1}^g X_i \otimes F_i \leq I \qquad \iff \qquad \bra{\psi}Z\ket{\psi} \geq 0 \quad \forall \ket{\psi} \in \mathbb C^d \otimes \mathbb C^n, \braket{\psi|\psi} = 1.
\end{equation*}
Now, we can use the Schmidt decomposition of $\ket{\psi} = \sum_{i=1}^d \lambda_i \ket{g_i} \otimes \ket{f_i}$, where $\lambda_i \geq 0$, $\sum_{i = 1}^d \lambda_i^2 = 1$, $\{\ket{g_i}\}_{i \in [d]}$ is an orthonormal basis of $\mathbb C^d$ and $\{\ket{f_i}\}_{i \in [d]}$ is a set of orthonormal vectors in $\mathbb C^n$. We define
\begin{equation*}
    Q = \sum_{i=1}^d \lambda_i \dyad{g_i} \qquad \mathrm{and} \qquad V = \sum_{i=1}^d \dyad{f_i}{g_i},
\end{equation*}
where $V$ is an isometry. Thus, we can write
\begin{equation*}
    \bra{\psi}Z\ket{\psi} = \bra{\Omega} (Q \otimes V)^\ast Z (Q \otimes V) \ket{\Omega},
\end{equation*}
where $\ket{\Omega} = \sum_{i = 1}^d \ket{g_i} \otimes \ket{g_i}$ is an unnormalized maximally entangled state. Finally, we note that 
\begin{equation*}
    (I \otimes V^\ast) Z (I \otimes V) = I - \sum_{i = 1}^g X_i \otimes V^\ast F_i V
\end{equation*}
Since $\mathcal F$ is a matrix convex set, it is closed under UCP maps, thus $(V^\ast F_1 V, \ldots, V^\ast F_g V) \in \mathcal F_d$. By assumption, thus $   (I \otimes V^\ast) Z (I \otimes V) \geq 0$. As $\ket{\psi}$ was arbitrary, $ \sum_{i = 1}^g X_i \otimes F_i \leq I$, and as $(F_1, \ldots, F_g)$ was also arbitrary, the assertion follows.
\end{proof}
With that, it follows that although we cannot necessarily write $\mathcal W^{\mathrm{min}}_d(\mathcal C)$ in the form of $\{X \in (\mathcal M(\mathbb C)_d^{\mathrm{sa}})^g: \sum_{i = 1}^g X_i \otimes A_i \leq I\}$ for some $A \in (\mathcal M(\mathbb C)_d^{\mathrm{sa}})^g$, it can be written as the intersection of such sets:
\begin{cor}
let $\mathcal C \subset \mathbb R^g$ be a closed convex set containing $0$. Then,
\begin{equation*}
    \mathcal W_d^{\mathrm{min}}(\mathcal C) = \bigcap_{F \in \mathcal W^{\mathrm{max}}_d(\mathcal C^\circ)} \mathcal D_F(d),
\end{equation*}
where $\mathcal D_F(d) = \{X \in (\mathcal M(\mathbb C)_d^{\mathrm{sa}})^g: \sum_{i = 1}^g X_i \otimes F_i \leq I\}$.
\end{cor} 
This corollary shows that Definition \ref{def:Delta} agrees with the definition for the inclusion set used in the previous work \cite{bluhm2018joint, bluhm2020compatibility}. The proof technique is very similar to \cite[Lemma 5.2]{bluhm2018joint}.

\section{Compatibility and tensor norms} \label{sec:compatibility-norms}

In this section, we relate the notion of measurement incompatibility in quantum mechanics to a tensor norm defined on a real vector space.

The starting point of our investigation is the following optimization problem. 

\begin{defi}\label{def:c-norm}
For a tensor $A \in \mathbb R^g \otimes \mathcal M(\mathbb C)^{\mathrm{sa}}_d$, we define $\|A\|_{\mathrm{c}}$ to be the value of the following optimization problem:
\begin{align}
\label{eq:c-norm-optimization}
    \mathrm{minimize} \quad & \quad \lambda_{\mathrm{max}}\left(\sum_j H_j\right) \\
\nonumber
    \mathrm{such~that} \quad & \quad A = \sum_j z_j \otimes H_j \\
\nonumber
    & \quad \|z_j\|_{\infty} = 1 \qquad \forall j\\
\nonumber
    & \quad H_j \geq 0 \qquad \forall j\\
\nonumber
    & \quad z_j \in \mathbb R^g, H_j \in \mathcal M(\mathbb C)_d^{\mathrm{sa}}  \qquad \forall j
\end{align}
\end{defi}

In the optimization problem above, we assume that the sums are finite, but can have arbitrary length. For this reason, as well as due to the tensor product between the variables $z_j$ and $H_j$, the problem is not stated in the form of a semidefinite program (SDP). However, it can be put in SDP form, as follows. 
\begin{align}
\label{eq:c-norm-SDP}
    \mathrm{minimize} \quad & \quad \lambda\\
\nonumber
    \mathrm{such~that} \quad & \quad A = \sum_{l=1}^{2^g} \epsilon_l \otimes K_l \\
\nonumber
    & \quad \lambda I_d \geq \sum_{l=1}^{2^g} K_l\\
\nonumber
    & \quad K_l \geq 0 \qquad \forall l \in [2^g]\\
\nonumber
    & \quad K_l \in \mathcal M(\mathbb C)_d^{\mathrm{sa}}  \qquad \forall l \in [2^g]
\end{align}
where $\{\epsilon_l\}_{l=1}^g$ is an enumeration of $\{\pm 1\}^g$, the extreme points of the unit ball of of the $\ell_\infty^g$ norm. Note that in the formulation above, we have only $2^g$ positive semidefinite variables, and the problem is in SDP form. To show that the former problem can be reduced to the latter form, decompose each $z_j$ as a convex combination of the
extreme points $\epsilon_l$:
$$z_j = \sum_{l=1}^{2^g} \mu(l|j) \epsilon_l,$$
where $\mu(\cdot | \cdot)$ is a conditional probability distribution. We can write now
$$A = \sum_j z_j \otimes H_j = \sum_{l=1}^{2^g} \epsilon_l \otimes \underbrace{\sum_{j} \mu(l|j) H_j}_{=:K_l}.$$
Note that $K_l \geq 0$ and 
$$\sum_{l=1}^{2^g} K_l = \sum_{l=1}^{2^g} \sum_j \mu(l|j) H_j = \sum_j H_j,$$
so the values of the two optimization problems are the same. 

This optimization problem actually defines a norm on $A \in \mathbb R^g \otimes \mathcal M(\mathbb C)^{\mathrm{sa}}_d$.

\begin{lem}
For any $g$, $d \in \mathbb N$, $\|\cdot\|_{\mathrm{c}}$ is a norm on $\mathbb R^g \otimes \mathcal M(\mathbb C)^{\mathrm{sa}}_d$.
\end{lem}
\begin{proof}
The fact that $\|A\|_{\mathrm{c}} \iff A = 0$ follows readily since $\lambda_{\mathrm{max}}\left(\sum_j H_j\right) = 0$ if and only $H_j = 0$ for all $j$ for a sum of positive semidefinite matrices. $\|cA\|_{\mathrm{c}}= c\|A\|_{\mathrm{c}}$ for $c \in \mathbb R_+$ follows from the fact that $\lambda_{\max}(cH) = c\lambda_{\max}(H)$ for a positive semidefinite matrix $H$. Subadditivity follows from the fact that for valid decompositions of $A$, $B \in \mathbb R^g \otimes \mathcal M(\mathbb C)^{\mathrm{sa}}_d$, their sum is a valid decomposition of $A + B$ and the fact that the operator norm is subadditive.
\end{proof}
Moreover, we can show that $\|\cdot\|_{\mathrm{c}}$ is a reasonable crossnorm, when endowing the $\mathbb R^g$ and $\mathcal M(\mathbb C)^{\mathrm{sa}}_d$ with their respective $\|\cdot \|_\infty$ Banach space norms.
\begin{prop}
For any $g$, $d \in \mathbb N$, $\|\cdot\|_{\mathrm{c}}$ is a reasonable cross norm on $\ell_\infty^g \otimes S_\infty^d$.
\end{prop}
\begin{proof}
We will show that $\|\cdot\|_{\epsilon} \leq \|\cdot\|_{\mathrm{c}} \leq \|\cdot\|_{\pi}$. The assertion will then follow from Proposition \ref{prop:reasonable-cross-norms}. Let $A \in \mathbb R^g \otimes \mathcal M(\mathbb C)^{\mathrm{sa}}_d$. Then,
\begin{equation*}
     \norm{A}_{\pi}= \inf \left\{\sum_i \norm{Y_i}_{\infty}: A = \sum_{i} z_i \otimes Y_i,~\|z_i\|_{\infty} = 1~\forall i  \right\}.
\end{equation*}
Since every Hermitian matrix $Y_i$ can be decomposed into a positive and a negative part as $Y_i = Y_i^+ - Y_i^-$, with $Y_i^+ Y_i^- = 0$, $Y_i^\pm \geq 0$, we can use
\begin{equation*}
    \lambda_{\max}\Big(\sum_i(Y_i^+ + Y_i^-)\Big) = \Big\| \sum_i |Y_i| \Big\|_{\infty} \leq \sum_i \| Y_i \|_{\infty} 
\end{equation*}
to infer $\|A\|_{\mathrm{c}} \leq \|A\|_{\pi}$. Now consider 
\begin{equation*}
          \norm{A}_{\epsilon}:= \sup \Big\{\big| \sum_j \phi(z_j) \psi(H_j)\big|: \norm{\phi}_{\ell_1^g} \leq 1, \norm{\psi}_{S_1^d} \leq 1   \Big\},
\end{equation*}
where $z_j$, $H_j$ define the decomposition achieving $\|A\|_{\mathrm{c}}$. Then, as $\|z_j\|_{\infty} = 1$,
\begin{equation*}
    \Big| \sum_j \phi(z_j) \psi(H_j)\Big| \leq \Big\|\sum_j \phi(z_j) H_j\Big\|_{\infty} = \Big\|\sum_j a_j H_j\Big\|_{\infty},
\end{equation*}
where $a_j \in  [-1,1]$. Now, as 
\begin{equation*}
    -\sum_{j:a_j < 0} H_j \leq \sum_j a_j H_j \leq \sum_{j:a_j > 0} H_j,
\end{equation*}
it holds that $\norm{A}_{\epsilon} \leq \|\sum_j H_j\|_{\infty}=\|A\|_{\mathrm{c}} $.
\end{proof}
Now we can establish a link between reasonable crossnorms and measurement compatibility: 
\begin{thm} \label{thm:comp-norm}
Let 
\begin{equation*}
    A = \sum_{j=1}^g e_j \otimes (2E_j - I).
\end{equation*}
Then,
\begin{enumerate}
    \item $\|A\|_{\ell^g_\infty \otimes_\epsilon S^d_\infty} \leq 1$ if and only if $\{E_j\}_{j  \in [g]}$ is a collection of effects.
    \item $\|A\|_{\mathrm{c}} \leq 1$ if and only if $\{E_j\}_{j  \in [g]}$ is a collection of \emph{compatible} effects.
\end{enumerate}
\end{thm}
\begin{proof}
Since the extreme points of $\ell^g_1$ are $\pm e_i$, $i \in [g]$, the first condition is equivalent to $\| 2 E_i - I\|_\infty \leq 1$ for all $i \in [g]$. This is easily seen to be equivalent to $0 \leq E_i \leq I$.

The expression for $\|A\|_{\mathrm{c}}$ implies in particular that for $z_j(i)$ the $i$-th coordinate in the standard basis, 
\begin{equation} \label{eq:effects}
    2 E_i - I = \sum_j z_j(i) H_j.
\end{equation}
Let $\|A\|_{\mathrm{c}} \leq 1$. Then, we can assume that the $H_j$ sum to the identity (possibly no longer considering the $\{H_j\}_j$ whose maximal eigenvalue is $\|A\|_{\mathrm{c}} \leq 1$), because otherwise we could just add $I-\sum_j H_j$ and assign $z_j = 0$ for that operator. In this case, it is easy to see that Eq.~\eqref{eq:effects} is equivalent to 
\begin{equation*}
    E_i = \sum_{j} \frac{1 + z_j(i)}{2} H_j
\end{equation*}
and
\begin{equation*}
    I - E_i = \sum_{j} \frac{1 - z_j(i)}{2} H_j.
\end{equation*}
Thus, $\{H_j\}_{j}$ is a joint POVM from which the $E_i$ arise by classical post-processing of the outcomes with conditional probabilities $p(\pm|i,j)$, where $p(\pm|i,j) = \frac{1 \pm z_j(i)}{2}$. Note that $p(\pm|i,j) \in [0,1]$ since $\norm{z_j}_\infty = 1$. Thus, $\{E_i\}_{i \in g}$ is a collection of compatible effects. For the reverse implication, we can use the joint POVM as $H_j$ and build the $z_j$ from the conditional probabilities for classical post-processing as above. Note that relaxing the requirement $\|z_j\| = 1$ to $\|z_j\| \leq 1$ does not change the value of $\|A\|_{\mathrm{c}}$. This shows that $\|A\|_{\mathrm{c}} \leq 1$, since the $H_j$ sum to the identity.
\end{proof}

\begin{remark}
We can also show easily from this formulation as a tensor norm that the post-processing (Lemma \ref{lem:post-processing}) and marginal (Definition \ref{def:jointPOVM}) points of view for compatibility are equivalent. To this end, we start from the SDP formulation of the compatibility norm, and notice that

\begin{equation*}
        E_i = \sum_{l=1}^{2^g} \frac{1 + \epsilon_l(i)}{2} K_l = \sum_{l:\epsilon_l(i)=1} K_l
\end{equation*}
and
\begin{equation*}
    I - E_i = \sum_{l=1}^{2^g} \frac{1 - \epsilon_l(i)}{2} K_l= \sum_{l:\epsilon_l(i)=-1} K_l,
\end{equation*}
since clearly, $\frac{1 - \epsilon_l(i)}{2} \in \{0,1\}$. So we have shown that the $K_l$ form a joint POVM from which the effects $E_i$ arise as marginals.
\end{remark}

\begin{remark}
The above results also follow from Theorem 9.2 and Proposition 9.4 of \cite{Bluhm2020GPT}, which more generally deal with general probabilistic theories, of which quantum mechanics is a special case. We presented detailed proofs here in order to make the presentation more accessible to the reader interested only in quantum theory.
\end{remark}

We relate next the compatibility norm $\|\cdot\|_{\mathrm c}$ to matrix convex sets, providing an alternative point of view to Theorem \ref{thm:comp-norm}.

\begin{prop} \label{prop:min-and-max}
Let $g$, $d \in \mathbb N$ and let $\mathcal C$ be the unit ball of $\ell_\infty^g$. Then, the unit ball of $\|\cdot\|_{\mathrm{c}}$ in $\mathbb R^g \otimes \mathcal M(\mathbb C)_d^{\mathrm{sa}} $ is equal to $\mathcal W^{\min}_d(\mathcal C)$. Moreover, the unit ball of $\|\cdot\|_{\ell^g_\infty \otimes_\epsilon S^d_\infty}$ is equal to $\mathcal W^{\max}_d(\mathcal C)$. 
\end{prop}
\begin{proof}
The first assertion follows immediately from the definition of $\|\cdot\|_{\mathrm{c}}$, since we can relax the constraint $\|z_j\|_\infty = 1$ to $\|z_j\|_\infty \leq 1$ without changing the value of the optimization problem. The second assertion follows, because for
\begin{equation*}
    A = \sum_{j = 1}^g e_j \otimes A_j,
\end{equation*}
$\|A\|_{\ell^g_\infty \otimes_\epsilon S^d_\infty} \leq 1$ if and only if $-I \leq A_j \leq I$ for all $j \in [g]$. The unit ball of $\ell_\infty^g$ is defined by hyperplanes $\{x \in \mathbb R^g: \mu x_i \leq 1\}$, $\mu \in \{\pm 1\}$, $i \in [g]$.
\end{proof}

Theorem \ref{thm:comp-norm} also allows us to identify the compatibility region from Definition \ref{def:Gamma} with the inclusion constant sets for the $\ell_\infty^g$ unit ball from Definition \ref{def:Delta}.

\begin{thm} \label{thm:gamma-is-cube}
Let $g$, $d \in \mathbb N$. Let $s \in [0,1]^g$. Then, $\{s_i E_i + (1-s_i)I/2\}_{i \in [g]}$ is a collection of compatible effects for all $g$-tuples $(E_i)_{i \in [g]} \in \mathrm{Eff}_d^g$, if and only if $s \in \Delta_{\square}(g,d)$. An equivalent way to phrase this is
\begin{equation*}
    \Gamma(g,d) = \Delta_{\square}(g,d).
\end{equation*}
\end{thm}
\begin{proof}
Let 
\begin{equation}
    A = \sum_{j = 1}^g e_j \otimes (2E_j - I) \label{eq:not-so-special-form}
\end{equation}
and 
\begin{equation*}
    A^\prime = \sum_{j = 1}^g e_j \otimes \left(2\left(s_j E_j + (1-s_j)\frac{I}{2}\right) - I\right). 
\end{equation*}
Then, $A^\prime = s \cdot A$, where the multiplication is understood component-wise. Now, by Theorem \ref{thm:comp-norm}, $\{s_i E_i + (1-s_i)I/2\}_{i \in [g]}$ is a collection of compatible effects if and only if $\|s \cdot A\|_{\mathrm{c}} \leq 1$. Thus, by Proposition \ref{prop:min-and-max}, $s \cdot A \in \mathcal W^{\min}_d(\mathcal C)$, where $\mathcal C$ is the unit ball of $\ell_\infty^g$. As $A \in \mathcal W^{\min}_d(\mathcal C)$ if $E_i \in \mathrm{Eff}_d$ for all $i \in [g]$ by Theorem \ref{thm:comp-norm} and Proposition \ref{prop:min-and-max}, we infer $s\cdot W^{\max}_d(\mathcal C) \subseteq W_d^{\min}(\mathcal C)$ and thus $s \in \Delta_{\square}(g,d)$. Conversely, if $s \in \Delta_{\square}(g,d)$, we can pick any $A \in \mathbb R^g \otimes \mathcal M(\mathbb C)_d^{\mathrm{sa}}$ and put it in the form of Eq.~\eqref{eq:not-so-special-form}. By Theorem \ref{thm:comp-norm} and Proposition \ref{prop:min-and-max}, $\|A\|_{\ell^g_\infty \otimes_\epsilon S^d_\infty} \leq 1$ implies that $\|s \cdot A\|_{\mathrm{c}} \leq 1$. Thus, $E_i \in \mathrm{Eff}_d$ for all $i \in [g]$ and $\{s_i E_i + (1-s_i)I/2\}_{i \in [g]}$ is a collection of compatible effects.
\end{proof}

\section{Incompatibility witnesses} \label{sec:witness-norms}
In the theory of entanglement, the notion of entanglement witnesses \cite[Section VI.B.3]{horodecki2009quantum} plays a crucial role, as it allows, given a bipartite quantum state, to certify its non-separability \cite{terhal2000bell}. Similar notions have been developed for detecting incompatibility of quantum measurements \cite{carmeli2019quantum,jencova2018incompatible,bluhm2020compatibility}, and have been generalized to general probabilistic theories \cite{kuramochi2020compact,Bluhm2020GPT}.

We have already seen that compatibility of measurements is related to studying tensors on $\mathbb R^g \otimes \mathcal M(\mathbb C)_d^{\mathrm{sa}}$, using reasonable crossnorms on $\ell^g_\infty \otimes S^d_{\infty}$. Thus, let us look at the duals of these spaces, which correspond again to the real vector space $\mathbb R^g \otimes \mathcal M(\mathbb C)_d^{\mathrm{sa}}$, but this time with reasonable crossnorms on $\ell^g_1 \otimes S^d_{1}$. 

Let us start by defining the set of \emph{effect witnesses} as the unit ball on $\ell^g_1 \otimes_\pi S^d_{1}$, which is the space dual to $\ell^g_\infty \otimes_\epsilon S^d_\infty$:
\begin{equation*}
    \mathcal E_d:= \{\varphi \in \mathbb R^g \otimes \mathcal M(\mathbb C)_d^{\mathrm{sa}}: \|\varphi\|_{\pi} \leq 1\} =  \Big\{\varphi \in \mathbb R^g \otimes \mathcal M(\mathbb C)_d^{\mathrm{sa}}: \sum_{i=1}^g\|X_i\|_{1} \leq 1\Big\}.
\end{equation*}
Here, we have written
\begin{equation*}
    \varphi = \sum_{i = 1}^g e_i \otimes X_i.
\end{equation*}
Since the projective and injective norm are dual to each other, it is straightforward to see that for
\begin{equation*}
    A = \sum_{i = 1}^g e_i \otimes (2 E_i - I),
\end{equation*}
the $E_i \in \mathcal M(\mathbb C)_d^{\mathrm{sa}}$ are effects if and only if 
\begin{equation*}
    \langle \phi, A \rangle = \Tr[\phi A] = \sum_{i = 1}^g \mathrm{Tr}[X_i(2E_i-I)] \leq 1
\end{equation*}
for all $\varphi \in \mathcal E_{d}$.

Moving now to compatible effects, one can compute the dual program to the SDP in Eq.~\eqref{eq:c-norm-SDP} giving the value of the compatibility norm $\|A\|_{\mathrm c}$:
\begin{align}
\label{eq:c-norm-dual-SDP}
    \mathrm{maximize} \quad & \quad \sum_{i=1}^g \mathrm{Tr}[\phi_i A_i] \\
\nonumber
    \mathrm{such~that} \quad & \quad \rho - \sum_{i=1}^g \epsilon_l(i) \phi_i \geq 0 \qquad \forall l \in [2^g] \\
\nonumber
    & \quad \Tr \rho = 1 \\
\nonumber
    &  \rho, \phi_i \in \mathcal M(\mathbb C)_d^{\mathrm{sa}}  \qquad \forall i \in [g].
\end{align}
Averaging the constraints, we can easily see that $\rho \geq 0$. Moreover, the equality constraint $\Tr \rho =1$ can be relaxed to $\Tr \rho \leq 1$. Since the dual SDP is strictly feasible (consider $\rho = I/(2d)$ and $\phi_i = 0$ for all $i \in [g]$), strong duality holds and the value of both SDPs is the same. Moreover, it is easy to see that the primal SDP is also strictly feasible, decomposing $A_i$ into positive and negative part, adding the identity if necessary to make both parts positive definite.

Now, let us consider $\|\cdot\|_{\mathrm{c*}}$, the dual norm of $\|\cdot\|_{\mathrm{c}}$. Since $\|\cdot\|_{\mathrm{c*}}$ is the dual norm of a reasonable crossnorm, it is a reasonable crossnorm itself. 

 \begin{prop}
 The norm $\|\cdot\|_{\mathrm{c}*}$, dual to the compatibility norm from Definition \ref{def:c-norm}, has the following expression: 
 $$\|\phi\|_{\mathrm c*} = \inf\Big\{ \Tr \rho : \forall \epsilon \in \{\pm 1\}^g, \, \rho \geq \sum_{i=1}^g \epsilon_i \phi_i\Big\}, \qquad \forall \phi = \sum_{i=1}^g e_i \otimes \phi_i \in \mathbb R^g \otimes \mathcal M(\mathbb C)_d^{\mathrm{sa}}.$$
 \end{prop}
 \begin{proof}
 Since the semidefinite programs in Eq.~\eqref{eq:c-norm-SDP} and Eq.~\eqref{eq:c-norm-dual-SDP} are dual and Eq.~\eqref{eq:c-norm-dual-SDP} is strictly feasible as noted above, they have the same value by Slater's condition, hence
 $$\|A\|_{\mathrm c} = \sup\{ \langle \phi , A \rangle \, : \, \|\phi\|_{\mathrm c*} \leq 1\},$$
 using the definition of the $\|\cdot\|_{\mathrm c*}$ quantity from the statement. But this is precisely the definition of norm duality, proving the claim.  
 \end{proof}

From the proposition above, it is clear that the unit ball of $\|\cdot\|_{\mathrm{c*}}$ is the set
\begin{equation}\label{eq:def-Id}
    \mathcal I_d := \Big\{\varphi = \sum_{i=1}^g e_i \otimes \varphi_i: \exists \rho \in \mathcal S(\mathbb C^d)~\mathrm{s.t.}~\rho - \sum_i \epsilon_i \varphi_i \geq 0~\forall \epsilon \in \{\pm1\}\Big\}.
\end{equation}
 Looking at the dual SDP, we see that the unit ball of $\|\cdot\|_{\mathrm{c}}$ is the polar of this set. We note that the above set is convex, closed and contains $0$, which concludes the proof.

 We call the set ${\mathcal I_d}$ the set of \emph{incompatibility witnesses}, because by duality $A \in \mathbb R^g \otimes \mathcal M(\mathbb C)^{\mathrm{sa}}_d$ corresponds to compatible effects if and only if
 \begin{equation*}
     \mathrm{Tr}[\varphi A] \leq 1 \qquad \forall \varphi \in {\mathcal I_d}.
 \end{equation*}
 Similar to entanglement witnesses, we are interested in the set of incompatibility witnesses which actually witness incompatibility for some collection of effects. We call these the \emph{strict incompatibility witnesses}, i.e.
 \begin{equation*}
     \mathcal{SI}_d:= \mathcal I_d \setminus \mathcal E_d.
 \end{equation*}
 In other words, $\phi \in \mathcal{SI}_d$ if, for all collection of \emph{compatible} effects $A$, $\langle \phi, A \rangle \leq 1$ but there exists a collection of (incompatible) effects $B$ such that $\langle \phi, B \rangle > 1$.
We sum up the discussion up to this point in the following proposition. 

\begin{prop}
A $g$ tuple of Hermitian, $d \times d$ complex matrices $\phi \in \mathbb R^g \otimes \mathcal M(\mathbb C)^{\mathrm{sa}}_d$ is called an \emph{effect witness} if, for all $g$-tuple of measurement operators $A$, $\langle \phi, A \rangle = \Tr[\phi A] \leq 1$. The set of effect witnesses, denoted by $\mathcal E_d$, is the unit ball of the $\ell^g_1 \otimes_\pi S_1^d$ norm.

Similarly, $\phi$ is called an \emph{incompatibility witness} if, for all $g$-tuple of \emph{compatible} measurement operators $A$, $\langle \phi, A \rangle = \Tr[\phi A] \leq 1$. The set of incompatibility witnesses, denoted by $\mathcal I_d$, is the unit ball of the $\ell^g_1 \otimes_{\mathrm c*} S_1^d$ norm.
\end{prop}

\bigskip

Up to this point, we have considered the duality relation measurement - effect given by the standard scalar product $\langle \cdot , \cdot \rangle$ in $\mathbb R^g \otimes \mathcal M(\mathbb C)^{\mathrm{sa}}_d$. We shall now consider duality in the \emph{matrix convex set} setting \cite{davidson2016dilations} (also called polar duality), where the condition 
$$\langle \phi , A \rangle = \sum_{i=1}^g \Tr[\phi_i A_i] \leq 1$$
is replaced by
$$\sum_{i=1}^g \phi_i \otimes A_i \leq I.$$
Importantly, under this duality, the maximal (resp.~minimal) matrix convex set corresponding to the $\ell^g_\infty$ ball corresponds to the minimal (resp.~maximal) matrix convex set of the $\ell^g_1$ ball, see Lemma \ref{lem:mcs-dual}. We shall consider the duals (in the sense of matrix convex sets) of the sets of effects and compatible effects from Proposition \ref{prop:min-and-max}.
 
We start the maximal matrix convex set for the unit ball of $\ell_1^g$. It is defined as
\begin{equation*}
    \mathcal W^{\max}_d(B_{\ell_1^g}) := \left\{X \in (\mathcal M(\mathbb C)^{\mathrm{sa}}_d)^g: \sum_{i = 1}^g \epsilon_i X_i \leq I \quad \forall \epsilon \in \{\pm 1\}\right\}.
\end{equation*}
It is again possible to express $\mathcal W^{\max}_d(B_{\ell_1^g})$ as the unit ball of a norm:
\begin{prop}
The unit ball of $\|\cdot\|_{\ell_1^g \otimes_\epsilon S^d_\infty}$ is $\mathcal W^{\max}_d(B_{\ell_1^g})$.
\end{prop}
\begin{proof}
This follows directly from Eq.~\eqref{eq:injective-norm} with $X = S_\infty^d$.
\end{proof}

To study the minimal matrix convex set associated to the unit ball of the $\ell_1^g$ norm, consider the following norm: 
    \begin{equation*}
        \|X\|_{\mathrm{wit}}:=\sup_{ \rho \in \mathcal S(\mathbb C^d)} \sum_{i = 1}^g \|\rho^{1/2} X_i \rho^{1/2}\|_1.
    \end{equation*}
Since the matrices $\rho$ above are quantum states, it is quite easy to see that the quantity above indeed defines a norm on $\mathbb R^g \otimes \mathcal M(\mathbb C)_d^{\mathrm{sa}}$.

\begin{prop}\label{prop:wit-is-min}
The unit ball of $\|\cdot\|_{\mathrm{wit}}$ is $\mathcal W^{\min}_d(B_{\ell_1^g})$.
\end{prop}
\begin{proof}
We start by showing that $\mathcal W_d^{\min}(B_{\ell_1^g})$ is contained in the unit ball of $\|\cdot\|_{\mathrm{wit}}$. Thus, for $A \in \mathcal W_d^{\min}(B_{\ell_1^g})$, we have that
\begin{equation*}
    A = \sum_{i=1}^g e_i \otimes (P_i - N_i),
\end{equation*}
where $P_i$, $N_i \geq 0$ for all $i \in [g]$, $\sum_{i=1}^g (P_i + N_i) = I$, and we have used the fact that the extreme points of $B_{\ell_1^g}$ are $\pm e_i$. Hence, 
    \begin{align*}
        \|A\|_{\mathrm{wit}}&=\sup_{ \rho \in \mathcal S(\mathbb C^d)} \sum_{i = 1}^g \|\rho^{1/2} (P_i - N_i) \rho^{1/2}\|_1 \\
        &\leq\sup_{ \rho \in \mathcal S(\mathbb C^d)} \sum_{i = 1}^g (\mathrm{Tr}[\rho P_i] + \mathrm{Tr}[\rho N_i])\\
        &\leq 1.
    \end{align*}
Let $\mathcal C$ be the unit ball of $\|\cdot\|_{\mathrm{wit}}$. For the reverse inclusion, we need to show that
\begin{equation*}
    \mathcal C \subseteq \mathcal W^{\min}_d(B_{\ell_1^g}) = \mathcal W^{\max}_d(B_{\ell_\infty^g})^\bullet,
\end{equation*}
where we have used the duality of matrix convex sets. Thus, we need to show that for $(X_1, \ldots, X_g) \in \mathcal C$,
\begin{equation*}
    \sum_{i = 1}^g X_i \otimes A_i \leq I
\end{equation*}
for all $A_i \in \mathcal M(\mathbb C)_n^{\mathrm{sa}}$ such that $\|A_i\|_\infty \leq 1$ for all $i \in [g]$, $n \in \mathbb N$. Equivalently, we can show that for all such $A_i$,
\begin{equation} \label{eq:sandwich}
    \sum_{i = 1}^g \bra{\psi}X_i \otimes A_i\ket{\psi} \leq 1 \qquad \forall \ket{\psi} \in \mathbb C^d \otimes \mathbb C^n, \|\psi\|_2 = 1.
\end{equation}
Let $\ket{\Omega} = \frac{1}{\sqrt{d}}\sum_{i = 1}^d \ket{e_i}\otimes\ket{e_i}$. For $n \geq d$, we can write
\begin{equation*}
    \ket{\psi} = \sqrt{d \rho_1} \otimes V \ket{\Omega},
\end{equation*}
where $V: \mathbb C^d \hookrightarrow \mathbb C^n$ is an isometry and $\rho_1$ is the reduced density matrix of $\ket{\psi}$ on $\mathbb C^d$. Thus, we can write Eq.~\eqref{eq:sandwich} as 
\begin{equation} \label{eq:sandwich2}
        d\sum_{i = 1}^g \bra{\Omega} \rho^{1/2}X_i \rho^{1/2} \otimes V^\ast A_i V \ket{\Omega}\leq 1 \qquad \forall \rho \in \mathcal S(\mathbb C^d).
\end{equation}
Using that $\bra{\Omega} A \otimes B\ket{\Omega} = 1/d \mathrm{Tr}[A^T B]$, we rewrite the left hand side of Eq.~\eqref{eq:sandwich2} as 
\begin{equation} \label{eq:sandwich3}
     \sum_{i = 1}^g \mathrm{Tr}[ (V^\ast A_i V)^T \rho^{1/2}X_i \rho^{1/2}] \leq \sum_{i = 1}^g \| \rho^{1/2}X_i \rho^{1/2}\|_1,
\end{equation}
where we have used the fact that $\|V^\ast A_i V\|_\infty \leq \|A_i\|_\infty$ and the duality of Schatten $p$-norms. Thus, the right hand side of Eq.~\eqref{eq:sandwich3} is indeed upper bounded by $1$ for all $\rho \in \mathcal S(\mathbb C^d)$ for $(X_1, \ldots, X_g) \in \mathcal C$. Since for $n \leq d$, we can always embed the $A_i$ in dimension $d$ by adding zeroes, the assertion follows.
\end{proof}

\begin{remark}\label{rem:wit}
There is a close relation between the set $X \in \mathcal W^{\max}_d(B_{\ell_1^g})$ and the set $\mathcal I_d$ of incompatibility witnesses from Eq.~\eqref{eq:def-Id}: every $X \in \mathcal W^{\max}_d(B_{\ell_1^g})$ gives rise to an incompatibility witnesses:
\begin{equation} \label{eq:witness-family}
  X \in \mathcal W^{\max}_d(B_{\ell_1^g}) \implies  (\rho^{1/2} X_1 \rho^{1/2}, \ldots, \rho^{1/2} X_g \rho^{1/2}) \in \mathcal I_d. 
\end{equation}
Moreover, it is easy to see that every incompatibility witness arises in this way. This establishes a relation between the two types of duality considered in this section (see also Table \ref{tbl:inclusions}). Note that $X$ gives rise to elements in $\mathcal{SI}_d$ via Eq.~\eqref{eq:witness-family} for some $\rho \in \mathcal S(\mathbb C^d)$ if $\|X\|_{\mathrm{wit}}> 1$.
\end{remark}

\bigskip

From Proposition \ref{prop:wit-is-min}, we can also obtain a second characterization of the compatibility region in Definition \ref{def:Gamma}, this time as the inclusion set of $\ell_1^g$. 
\begin{thm} \label{thm:gamma-is-diamond}
Let $g$, $d \in \mathbb N$. Let $s \in [0,1]^g$. Then, $\{s_i E_i + (1-s_i)I/2\}_{i \in [g]}$ is a collection of compatible effects for all $g$-tuples $(E_i)_{i \in [g]} \in \mathrm{Eff}_d^g$, if and only if $s \in \Delta_{\diamond}(g,d)$. An equivalent way to phrase this is
\begin{equation*}
    \Gamma(g,d) = \Delta_{\diamond}(g,d).
\end{equation*}
\end{thm}

\begin{proof}
Let $A = \sum_{i = 1}^g e_i \otimes (2E_i -I)$ with $(E_i)_{i \in [g]} \in \mathrm{Eff}_d$ and $\varphi = \sum_{i=1}^g e_i \otimes (\rho^{1/2}X_i\rho^{1/2})$ for $X \in  \mathcal W^{\max}_d(B_{\ell_1^g})$ and $\rho \in \mathcal S(\mathbb C^d)$. We note that for $s \in [0,1]$, 
\begin{equation*}
    \mathrm{Tr}[s.\varphi A] =   \mathrm{Tr}[\varphi A^\prime],
\end{equation*}
where $A^\prime = \sum_{i = 1}^g e_i \otimes (2 (s_iE_i+(1-s_i)I/2 -I)$. Thus, if $s \in \Gamma(g,d)$, then $\mathrm{Tr}[s.\varphi A] \leq 1$ using Remark \ref{rem:wit}, which implies $\|s.X\|_{\mathrm{wit}} \leq 1$ since  $(E_i)_{i \in [g]} \in \mathrm{Eff}_d$ and $\rho \in \mathcal S(\mathbb C^d)$ were arbitrary. Thus, by Proposition \ref{prop:wit-is-min}, we infer that $s.X \in \mathcal W^{\min}_d(B_{\ell_1^g})$ and hence $s \in \Delta_{\diamond}(g,d)$. Conversely, let $s \in \Delta_{\diamond}(g,d)$. Then, $\mathrm{Tr}[\varphi A^\prime] \leq 1$ for all $\varphi \in {\mathcal I_d}$, since they can all be written as in Eq.~\eqref{eq:witness-family} and they only contain strict incompatibility witnesses if $X \not \in \mathcal W^{\min}_d(B_{\ell_1^g})$ by Proposition \ref{prop:wit-is-min}. Therefore, $A^\prime$ corresponds to compatible effects and $s \in \Gamma(g,d)$.
\end{proof}

\begin{remark}
From Theorems \ref{thm:gamma-is-cube} and \ref{thm:gamma-is-diamond}, it follows in particular that 
\begin{equation*}
    \Delta_{\diamond}(g,d) = \Delta_{\square}(g,d).
\end{equation*}
We could have proven this statement directly using duality of matrix convex sets as in \cite[Proposition 5.2]{bluhm2022steering}.
\end{remark}
 
\section{Discussion} \label{sec:final-section} 
In this paper, we have characterized quantum effects and compatible quantum effects with the help of tensor norms. We have done the same thing for the dual notions of effect and incompatibility witnesses. The notion of witness is defined via duality, and we have considered here duality in the usual sense of convex geometry, and in the matrix convexity sense. We summarize our findings in Table \ref{tbl:inclusions}. Note that the norms appearing are very natural, involving for the most part the usual $\ell_p$ norms on $\mathbb R^g$ and the Schatten classes on the space of Hermitian matrices.

\begin{table}[ht]
\bgroup
\def\arraystretch{2}
\begin{tabular}{ccc}
$\operatorname{Ball}(\ell_1^g \otimes_\pi S_1^d)=\mathcal E_d$ & $\subseteq$ & $\operatorname{Ball}(\ell_1^g \otimes_{\mathrm c*} S_1^d)=\mathcal I_d$ \\
$\Bigg\updownarrow$ & $\circ$ $-$ duality & $\Bigg\updownarrow$ \\
$\operatorname{Ball}(\ell_\infty^g \otimes_\epsilon S_\infty^d) = \mathcal W^{\max}_d(\operatorname{Ball}(\ell_\infty^g))$ & $\supseteq$ & $\operatorname{Ball}(\ell_\infty^g \otimes_{\mathrm c} S_\infty^d) = \mathcal W^{\min}_d(\operatorname{Ball}(\ell_\infty^g))$ \\
$\Bigg\updownarrow$ & $\bullet$ $-$ duality & $\Bigg\updownarrow$ \\
$\operatorname{Ball}(\|\cdot\|_{\mathrm{wit}})= \mathcal W^{\min}_d(\operatorname{Ball}(\ell_1^g))$ & $\subseteq$ & $\operatorname{Ball}(\ell_1^g \otimes_\epsilon S_\infty^d) = \mathcal W^{\max}_d(\operatorname{Ball}(\ell_1^g))$
\end{tabular}
\egroup
\medskip
\caption{Different perspectives on dichotomic quantum measurements and their witnesses. The middle row of the diagram describes the set of $g$-tuples of dichotomic measurements (left) and $g$-tuples of \emph{compatible} measurements (right). The top row describes the duals of these sets (in the usual, scalar product sense), while the bottom row describes their matrix convex duals. The top-right and bottom-left cells are related, see Remark \ref{rem:wit}.}
\label{tbl:inclusions}
\end{table}

\bigskip

Another focal point of our paper is the computation of the incompatibility region $\Gamma(g,d)$, for a given number $g$ of dichotomic measurements, and a given Hilbert space dimension $d$. We have established in this work that 
\begin{equation*}
    \Delta_{\diamond}(g,d) = \Gamma(g,d) = \Delta_{\square}(g,d)
\end{equation*}
(see Theorems \ref{thm:gamma-is-cube} and \ref{thm:gamma-is-diamond}). As discussed in \cite{bluhm2018joint} for the matrix diamond and in \cite{bluhm2022steering} for the matrix cube, these findings enable us to give concrete bounds on $\Gamma(g,d)$. In order to keep the present article self-contained, we will give in the following a concise discussion of these bounds. For the inclusion set of the matrix cube, the following proposition summarizes the findings in \cite{passer2018minimal, bluhm2018joint}:
\begin{prop} \label{prop:quarter-circle}
Let $g$, $d \in \mathbb N$ and let
\begin{equation*}
    \mathrm{QC}_g := \left\{s \in [0,1]^g: \sum_{i=1}^g s_i^2 \leq 1\right\}.
\end{equation*}
Then, $\mathrm{QC}_g \subseteq \Delta_\diamond(g,d)$. For $d \geq d^{\left\lceil \frac{g-1}{2} \right\rceil}$, it holds moreover that $\mathrm{QC}_g = \Delta_\diamond(g,d)$.
\end{prop}

For the matrix cube, we can summarize the results of \cite{bluhm2022steering}, which build on \cite{ben-tal2002tractable, helton2019dilations}, as follows:
\begin{prop}
Let $d \in \mathbb N$. The largest $\tau(d)$ such that $\tau(d)(1, \ldots, 1) \in \Delta_\square(g,d)$ for all $g \in \mathbb N$ is 
\begin{equation*}
    \tau_\ast(d) = 4^{-n} \binom{2n}{n}, \qquad n = \left\lfloor\frac{d}{2}\right\rfloor.
\end{equation*}
Asymptotically, $\tau_\ast(d)$ behaves as $\sqrt{2/(\pi d)}$.
\end{prop}

In particular, in the case of qubits ($d=2$), we find that $\tau_*(2) = 1/2$. This recovers a result from \cite{Bluhm2020GPT}, which has been obtained using a connection to $1$-summing norms of $\ell_2$ Banach spaces. We present our knowledge concerning $\Gamma(g,d)$ in Figure \ref{fig:phase-diagram}. The red region is the region in which $\tau_\ast(d) \not \in \mathrm{QC}_g$. Since we know that $\tau_\ast(d) \in \Gamma(g,d)$, it must thus hold that $\mathrm{QC}_g \subsetneq \Gamma(g,d)$. We see that we can recover all the results from Proposition \ref{prop:gamma-known-before} for $g=2$ and $g=3$, $d=2$.

\begin{figure}
    \centering
    \includegraphics{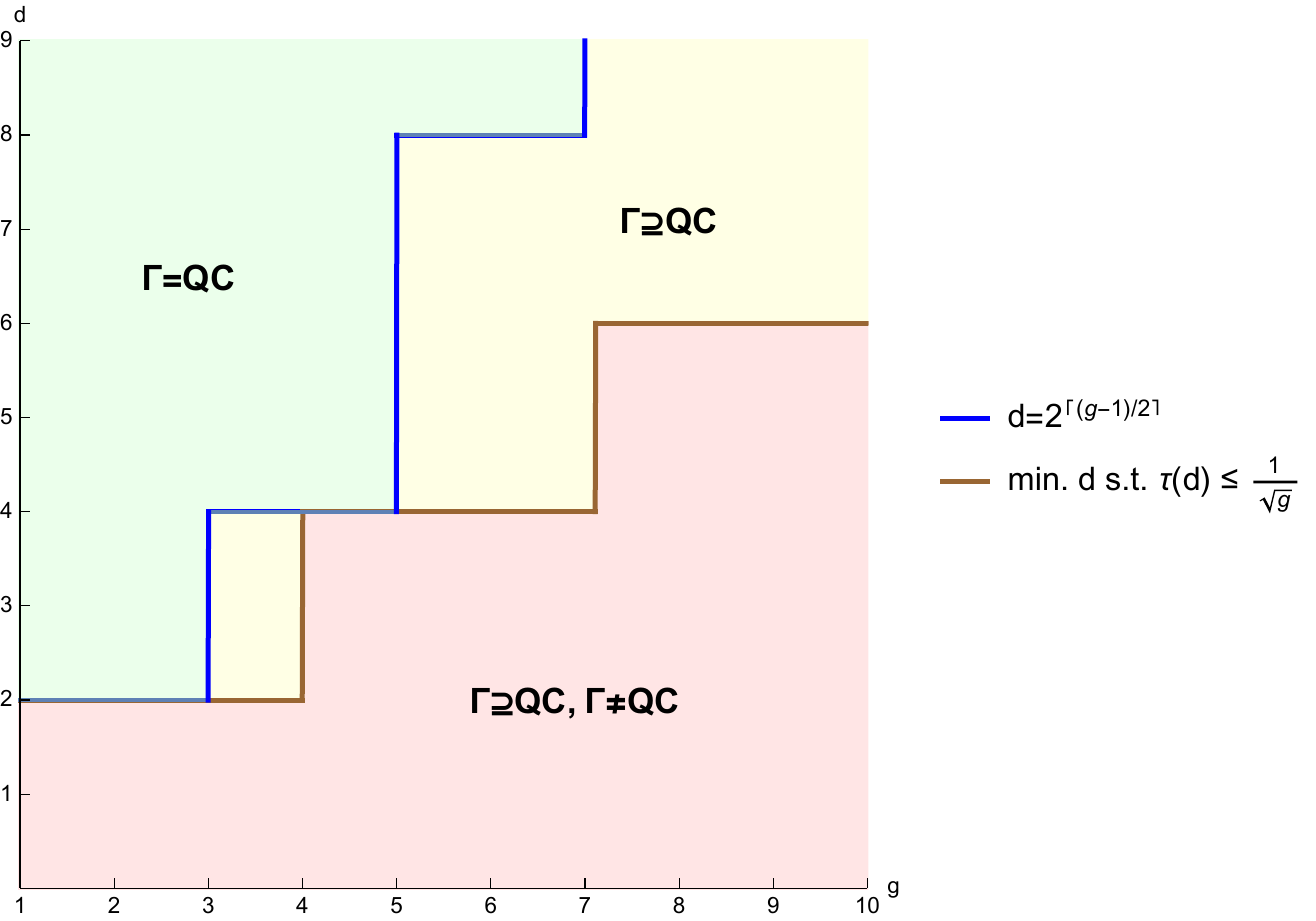}
    \caption{Bounds on the set $\Gamma(g,d)$. The set $\mathrm{QC}_g$ is defined as in Proposition \ref{prop:quarter-circle}. The blue line belongs to the green region in which $\Gamma(g,d) = \mathrm{QC}_g$, whereas the brown curve does not belong to the red region in which $\Gamma(g,d) \supsetneq \mathrm{QC}_g$}
    \label{fig:phase-diagram}
\end{figure}

\bigskip
 
To conclude, let us ask the question of determining compatibility regions and metric characterization of compatibility in the scenario where the measurements have more than two outcomes. This is a fundamental question of great importance in quantum theory, especially in the case of von Neumann measurements. In this paper we have considered tensor norms to describe dichotomic measurements. This approach cannot be extended to more general situations, because the set of measurements is no longer centrally symmetric and thus cannot be described as the unit ball of some norm. In our past work \cite{bluhm2020compatibility}, we have used the theory of free spectrahedra to circumvent this problem, but the corresponding mathematical machinery is still in its infancy and one needs to develop it further in order to obtain interesting results about the respective compatibility regions.

\bigskip

\textbf{Acknowledgements:}
A.B.~acknowledges financial support from the European Research Council (ERC Grant Agreement No. 81876) and VILLUM FONDEN via the QMATH Centre of Excellence (Grant No.10059). I.N.~was supported by the ANR project ``ESQuisses'' (grant number ANR-20-CE47-0014-01).

\bibliographystyle{alpha}
\bibliography{spectralit}
\end{document}